\documentclass[11pt]{article}
\usepackage{amsmath,amssymb,amsthm,braket,fullpage}
\usepackage{hyperref}
\usepackage{authblk}
\usepackage{appendix}

\newtheorem{theorem}{Theorem}

\newtheorem{proposition}{Proposition}
\newtheorem{corollary}{Corollary}
\theoremstyle{definition}
\newtheorem{remark}{Remark}
\newtheorem{definition}{Definition}

\title{Embezzlement as a ``Self-Test'' for Infinite Copies of Entangled States}
\author{Li Liu}
\affil{University of Copenhagen}

\date{}

\begin{document}

\maketitle
\begin{abstract}
We investigate the operator-algebraic structure underlying entanglement embezzlement, a phenomenon where a fixed entangled state (the catalyst) can be used to generate arbitrary target entangled states without being consumed. We show that the ability to embezzle a target state $g$ imposes strong internal constraints on the catalyst state $f$: specifically, $f$ must contain infinitely many mutually commuting, locally structured copies of $g$. This property is formalized using C*-algebraic tools and is analogous to a form of self-testing, certifying the presence of infinite copies of $g$ within $f$. Our results clarify the structural requirements for embezzlement and provide new conceptual tools for analyzing state certification in infinite-dimensional quantum systems.
\end{abstract}
\section{Introduction}

Embezzlement of entanglement is a very striking phenomena in quantum information: a special ``catalyst'' state allows Alice and Bob to locally create entanglement seemingly for free, without consuming the catalyst. Since its discovery by van Dam and Hayden~\cite{vandam2003embezzling}, embezzlement has played a role in nonlocal games\cite{coladangelo2020two, ji2020three}, complexity theory\cite{leung2013coherent}, and the study of entanglement as a resource\cite{zanoni2024complete}. The original construction by van Dam and Hayden was formulated in finite dimensions, but more recent developments in infinite-dimensional and operator-algebraic settings have revealed deeper structural phenomena\cite{van2024embezzlement, van2025multipartite}. These are particularly relevant when analyzing the ultimate constraints on quantum correlations.

To understand embezzlement better, a natural question is what structural features a state must have in order to function as an embezzler. In recent work~\cite{van2024embezzlement}, it was shown that any universal embezzler must generate a Type~III$_1$ von Neumann factor. In this work, we show something more concrete about the structure an embezzler must have, in terms of states it must contain rather than the type of von Neumann algebra it must be.

Most of the previous studies, including the original paper by van Dam and Hayden, consider embezzlement in the \textbf{approximate regime}. In this work, we focus instead on \textbf{exact embezzlement of entanglement}, which is more restrictive than the approximate case. All of our results concern exact embezzlement, where the catalyst remains completely unchanged and the target state is obtained exactly as specified using local operations.

We start by altering the embezzlement protocol by removing the separable state $\ket{00}$ from the input, and show that this new no-input embezzlement is equivalent to the original embezzlement protocol. By removing the $\ket{00}$, we perform isomorphisms on the catalyst state instead of unitary operations. It allows us to define isomorphisms that maps the catalyst state to the catalyst tensor infinite copies of the target state without running into issues of orthogonality of Hilbert spaces due to reference vector mismatches if we keep the $\ket{00}$ inputs.

We show that the ability to embezzle a given target state $g$ imposes strong internal constraints on the catalyst state $f$. Specifically, if $f$ can be used to embezzle $g$, then $f$ must contain infinitely many mutually commuting copies of $g$, with $g$'s locality preserved. This is similar to the idea of self-testing. A self-test certifies the presence of a specific entangled state by analyzing correlations produced by local measurements, without assuming anything about the internal workings of the devices. In our case, we use embezzlement to certify the existence of some infinite-copies of state $g$ within the catalyst $f$. However, unlike traditional self-testing, this approach does not rely on any specific set of observables or measurement statistics. We ground our certification based on the ability to perform exact embezzlement mathematically, which in itself does not contribute an experiment that can be performed to verify the procedure. In \cite{leung2013coherent}, a non-local game was proposed to verify the embezzlement of a Bell state. We could potentially use this non-local game as a verification procedure to certify the existence of infinite copies of Bell states within the catalyst. However, it is not clear in general how to certify universal embezzlement protocols through experiments.

Furthermore, in recent paper \cite{liu2025explicit}, an approximate universal embezzlement protocol was proposed using a catalyst state that consists of tensor product of all infinite-copies of entangled states with rational Schmidt coefficients. One criticism about such embezzler is that the state seems unnecessarily large. We show that with the main results from this paper, in order to have an embezzlement protocol that can exactly embezzle a dense set of states with all rational Schmidt coefficients, it is necessary for the catalyst to have the structure proposed in \cite{liu2025explicit}. However, it is not clear if similar result can be extended to the proposed exact universal embezzler for all states, due to the technical challenges in taking the tensor product of uncountably many states simultaneously.

The main contributions are as follows:
\begin{itemize}
\item We establish the equivalence between standard and "no-input" models of exact embezzlement in the infinite-dimensional setting, giving a precise operator-algebraic formulation.
\item We prove that exact embezzlement requires the presence of infinitely many disjoint copies of the target state within the catalyst, and formalize this using C*-algebraic tools.
\item We show that in order for a catalyst to embezzle any states with rational Schmidt coefficients exactly, it must contain tensor product of infinite copies of all those states simultaneously.
\end{itemize}

Beyond exact embezzlement, our results raise new questions about robustness:
is approximate embezzlement also a ``robust self-test'' for infinite copies of
entanglement? Developing such a theory would connect directly to central
themes in nonlocal games and quantum cryptography

\section{No-Input Embezzlement}
We start by showing that to perform embezzlement, the input of separable state $\ket{00}$ is not necessary. Following the proof of equivalences betweenstandard and no-input embezzlement, we describe no-input embezzlement using the C*-algebraic model, and explain why it is more natural than the commuting operator model. As a result, the rest of the paper will adopt the C*-algebraic model instead of the commuting operator model.
\subsection{Equivalence between Standard and No-input Embezzlement}
In entanglement embezzlement~\cite{vandam2003embezzling}, a fixed bipartite state $ \ket{\psi_c} \in \mathcal{H}_A \otimes \mathcal{H}_B $ is used as a catalyst to simulate or extract a target entangled state $ \ket{\psi} \in \mathbb{C}^n \otimes \mathbb{C}^n $, without consuming $ \ket{\psi_c} $. In the standard formulation, this is expressed as:
\[
\ket{\psi_c} \otimes \ket{00} \longmapsto \ket{\psi_c} \otimes \ket{\psi},
\]
where $ \ket{\psi_c} $ is the catalyst state and $ \ket{00} $ is the separable input state.

However, in many infinite-dimensional constructions, we consider an alternative formulation known as \emph{no-input embezzlement}, where the protocol directly maps
\[
\ket{\psi_c} \longmapsto \ket{\psi_c} \otimes \ket{\psi},
\]
without requiring any ancillary input. This variation simplifies many constructions and aligns with operator-algebraic models where explicit tensoring with finite ancillas is avoided. The original van Dam and Hayden construction also uses this no-input version of embezzlement. In \cite{van2024embezzlement}, the authors also showed that the two versions of embezzlements are equivalent to each other. We note that their paper is showing equivalence of embezzlement in the approximate case whereas our result is in the exact case.

We begin by showing that these two notions of embezzlement are equivalent.

\begin{proposition}[Equivalence of Embezzlement Models]\label{prop:equivalence_embezzlement}
Let $ \ket{\psi} \in \mathbb{C}^n \otimes \mathbb{C}^n $ be an entangled state. The following two notions of entanglement embezzlement are equivalent:

\begin{enumerate}
    \item \textbf{Standard Embezzlement:} There exist a catalyst state $ \ket{\psi_c} \in \mathcal{H} $ and local isometries
    \[
    U_A, U_B: \mathcal{H} \otimes \mathbb{C}^n \to \mathcal{H} \otimes \mathbb{C}^n
    \]
    such that
    \[
    (U_A\otimes 1_n)(U_B\otimes 1_n) ( \ket{\psi_c} \otimes \ket{00} ) = \ket{\psi_c} \otimes \ket{\psi}.\footnote{Here, $U_A \otimes 1_n$ and $U_B \otimes 1_n$ denote local isometries where the $1_n$ acts on the other party's single-qubit ancilla. For example, in $U_B \otimes 1_n$, the $1_n$ acts on the first qubit of $\ket{00}$. For notational simplicity, we write $U_B \otimes 1_n$ without specifying the precise ordering of tensor factors.}
    \]

    \item \textbf{No-Input Embezzlement:} There exist a catalyst state $ \ket{\phi_c} \in \mathcal{H} $ and local isometries
    \[
    V_B: \mathcal{H} \to \mathcal{H} \otimes \mathbb{C}^n, \quad V_A: \mathcal{H} \to \mathcal{H} \otimes \mathbb{C}^n,
    \]
    such that
    \[
    (V_A \otimes 1_n)(V_B \ket{\phi_c}) = \ket{\phi_c} \otimes \ket{\psi},
    \]
    and the isometries $ V_A, V_B $ satisfy
    \[
    (V_A \otimes 1_n) V_B = (V_B \otimes 1_n) V_A.\quad \footnote{This equality should be interpreted up to a swap of the two output qubit registers. Formally, the correct relation is $(V_A \otimes 1_n)(V_B) = S \circ (V_B \otimes 1_n)(V_A)$, where $S$ swaps the two qubit outputs. The expression $(V_A \otimes 1_n)V_B$ produces the output in the correct logical order: first qubit from $V_A$, second from $V_B$. For notational simplicity, we omit the swap and write equality without it, trusting that the reader understands the underlying ordering convention.}
    \]
\end{enumerate}

Either form of embezzlement can be converted into the other via explicit local isometries, with the catalyst space extended by ancillary qubits in a product state $ \ket{00}^{\otimes \infty} $ and local regrouping and swapping operations.
\end{proposition}

The proof of the proposition is rather technical and tedious and we will leave the details in the Appendix. The general idea is the following.

From standard to no-input, starting with $U_A$, $U_B$ and $\ket{\psi_c}$, we define $\ket{\phi_c} = \ket{\psi_c}\otimes \ket{00}^\infty$. Using the isomorphism $\ket{00}^\infty \simeq \ket{00}\otimes\ket{00}^\infty$, we construct our $V_A$ and $V_B $ accordingly so that a $\ket{00}$ is first generated from $\ket{00}^\infty$, $U_A$ and $U_B$ are applied, then finally and a swap is performed at the end to move $\ket\phi$ to the last register.

Similarly, from no-input to standard, we start with $V_A$, $V_B$ and $\ket{\phi_c}$. We define $\ket{\psi_c}$ to be $\ket{\phi_c}\otimes \ket{00}^\infty$. We then make use of the inverse isomorphism $\ket{00}\otimes\ket{00}^\infty \simeq \ket{00}^\infty$ to push the extra $\ket{00}$ back into $\ket{00}^\infty$, apply $V_A$ and $V_B$ then swap the resulting $\ket\phi$ to the last register. 

\subsection{C*-algebraic Formulation of No-Input Embezzlement}

In the rest of the paper, we will replace the standard commuting operator model for no-input embezzlement with the C*-algebraic model. This allows us to talk about no-input embezzlement more naturally. In particular, when we take the limit of the no-input embezzlement protocol to show the self-testing like properties of the catalyst, the C*-algebraic model captures the locality of the system more naturally than  the commuting operator model. In fact, the commuting operator model runs into some major issues which we will describe below. More discussions about the C*-algebraic model can be found in \cite{cleve2022constant,liu2025explicit}.

In the C*-algebraic model of no-input embezzlement, the protocol is described as a pair of *-isomorphisms $\Phi_A, \Phi_B$ acting on C*-algebras $\mathcal A\otimes \mathbb M_2 \to \mathcal A$ and $\mathcal B\otimes \mathbb M_2 \to\mathcal B$, and a state $f \in\mathcal S(\mathcal A\otimes\mathcal B)$ such that the following pullback exists: 
\[ f \circ (\Phi_A \otimes \Phi_B)= f \otimes g, \]
where $ g $ is the target state on $ \mathbb M_2 \otimes \mathbb M_2 $. This formulation allows us to better describe local isomorphisms with the structure of tensor product between them. As shown in \cite{liu2025explicit}, this C*-model naturally encodes exact embezzlement in the infinite-dimensional setting. Moreover, by applying the GNS construction to the state $ f $, this model can be turned into the commuting operator framework, recovering the standard interpretation.

Later on, when we discuss infinite copies of target state certification, the system requires us to map $\mathcal F_A\otimes\mathcal F_B\to\mathcal G_A\otimes\mathcal G_B$ where the catalyst state lies in $\mathcal S(\mathcal F_A\otimes\mathcal F_B)$ and the infinite copies of target state lies in $\mathcal S(\mathcal G_A\otimes\mathcal G_B)$. This map is achived by local *-homomorphisms, 
\[\pi_A:\mathcal G_A\to\mathcal F_A \quad \pi_B:\mathcal G_B\to\mathcal F_B.\]
The overall *-homomorphism $\pi = \pi_A\otimes \pi_B$ has bipartite local structure.

\subsubsection*{Issues with Commuting Operator Model}
If we attempt to describe the above system in the commuting-operator model, we immediately run into difficulties. The starting system is a Hilbert space $\mathcal H_f$ carrying commuting C$^*$-subalgebras $\mathcal F_A,\mathcal F_B\subseteq\mathbb B(\mathcal H_f)$; the certification system is $\mathcal H_g$ with commuting subalgebras $\mathcal G_A,\mathcal G_B\subseteq\mathbb B(\mathcal H_g)$. Let $\ket{\psi_f}\in\mathcal H_f$ be the starting state and $\ket{\psi_g}\in\mathcal H_g$ the intended certified state.

The ``local actions'' are specified in the Heisenberg picture by injective $*$-homomorphisms
\[
\theta_A:\mathcal G_A\hookrightarrow \mathcal F_A,\qquad 
\theta_B:\mathcal G_B\hookrightarrow \mathcal F_B.
\]
Since $\theta_A,\theta_B$ are $*$-monomorphisms, they admit spatial implementations by isometries
\[
V_A:\mathcal H_g\to\mathcal H_f,\qquad V_B:\mathcal H_g\to\mathcal H_f
\]
such that
\[
\theta_A(x)=V_A\,x\,V_A^*\quad(x\in\mathcal G_A),\qquad
\theta_B(y)=V_B\,y\,V_B^*\quad(y\in\mathcal G_B).
\]

If $\mathcal H_f$ and $\mathcal H_g$ had tensor-product decompositions, locality could be encoded at the state level by maps of the form $V_A\otimes V_B$, making the requirement ``Alice then Bob = Bob then Alice'' well-defined on $\ket{\psi_g}$. In the current setting, no such tensor structure is available. Both $V_A$ and $V_B$ map $\mathcal H_g$ into $\mathcal H_f$, so one cannot compose them in either order on vectors in $\mathcal H_g$. Thus, while operator-level locality (i.e., $\theta_A(\mathcal G_A)$ and $\theta_B(\mathcal G_B)$ landing in commuting subalgebras of $\mathbb B(\mathcal H_f)$) is well-posed, the naive vector-level notion of two commuting ``local maps'' $\mathcal H_g\to\mathcal H_f$ is ill-typed without additional structure.

\paragraph{Conclusion.}
When the local transformations are modeled as $*$-monomorphisms 
$\theta_A:\mathcal G_A\hookrightarrow \mathcal F_A$ and 
$\theta_B:\mathcal G_B\hookrightarrow \mathcal F_B$, the natural implementation by isometries 
$V_A,V_B:\mathcal H_g\to\mathcal H_f$ breaks down. At the state level, both $V_A$ and $V_B$ act from $\mathcal H_g$ into $\mathcal H_f$, so their 
compositions do not define a meaningful notion of ``Alice then Bob equals Bob then Alice'' on $\ket{\psi_g}$. 
By contrast, in the Heisenberg picture the situation is coherent: one simply requires the ranges of 
$\theta_A(\mathcal G_A)$ and $\theta_B(\mathcal G_B)$ to commute inside $\mathbb B(\mathcal H_f)$, 
which reduces the locality requirement to the existence of a single injective $*$-homomorphism
\[
\theta_A\otimes\theta_B:\mathcal G_A\otimes \mathcal G_B\;\hookrightarrow\;\mathcal F_A\otimes \mathcal F_B.
\]
Thus, for local maps with domain $\mathcal H_g$ and range $\mathcal H_f$, the C*-algebraic model provides a simple and correct formulation of the problem.
\begin{remark}
    We note that while directly applying the commuting operator model fails in this case, it can be achieved using von Neumann algebras through what Alain Conne's called ``correspondence / bimodules'' \cite[Ch.~5 Appendix~B] {connes1994ng}. It is also explained in \cite{crann2025operator}.
\end{remark}

\section{Embezzlement and Certification of Infinite Target States}
The C*-algebraic framework gives us the tools necessary to describe the ``self-test'' like properties of embezzlement of entanglement. We start by informally describing what it means for an infinite-dimensional state $ f \in \mathcal{S}(\mathcal{B}) $ to contain infinitely many copies of a state $ g \in \mathcal{S}(\mathcal{A}) $ with respect to the algebra of observables. If $ f $ admits $ n $ disjoint collections of observables, each generating the algebra that $ g $ acts on, and these collections are mutually commuting, then we may say that $ f $ contains $ n $ copies of $ g $. If this holds for all $ n \in \mathbb{N} $, we say that $ f $ contains infinitely many copies of $ g $. To claim $f$ contains copies of $g$ that respects its locality structure, we further require that each collection of observables respects the locality across different parties. In other words, each set can be partitioned into two mutually commuting subsets, one contained in Alice's algebra and one in Bob's.

Unless stated otherwise, throughout this paper we take the tensor product between Alice’s and Bob’s C*-algebras to be the maximal tensor product, as this is the standard construction used in the description of non-local quantum systems. In contrast, when forming tensor products between multiple copies of C*-algebras associated with copies of certain states, we use the minimal (spatial) tensor product. This choice ensures that the copies remain disjoint and that the norm reflects their disjointness. We remark that whenever the C*-algebras under consideration are nuclear, the minimal and maximal tensor products coincide. In particular, in the embezzlement of a single $n$-qubit state $g$, the C*-algebra generated by infinitely many copies of $g$ is the CAR algebra, which is nuclear. Consequently, in the context of this paper, the distinction between minimal and maximal tensor products is a technical matter rather than a substantive one.
\subsection{Definition of Infinite State Containment}

\begin{definition}[State Containment.]
\label{def:state_containment}
    Let $ g \in \mathcal S(\mathcal G)$ be a state on a C*-algebra $ \mathcal G $, and let $ f \in \mathcal S(\mathcal F) $ be a state on a C*-algebra $ \mathcal F $. We say that \emph{$f$ contains $g$} (up to unitary conjugations) if there exists a *-monomorphism $\pi : \mathcal G \to\mathcal F$ such that for all $A \in\mathcal G$, 
    \[f\circ\pi(A) = g(A).\]
\end{definition}

This definition requires $f$ to not only be able to reproduce the statistical distribution of $g$, but also to be able to fully mimic the behaviour of $g$ on the entire algebra $\mathcal A$. We note that, however, since conjugation by unitary produces a *-isomorphism, the containment only certifies $g$ up to unitary conjugation. While this might appear to make the definition weaker, because any pure states can be converted to each other by unitary conjugation, this definition will make more sense when we talk about local containment in our next definition. 

\begin{definition}[Local Containment]\label{def:local_containment}
Let $\mathcal F_A, \mathcal F_B, \mathcal G_A, \mathcal G_B$ be C*-algebras. Let $f\in \mathcal S(\mathcal F_A\otimes\mathcal F_B)$ and $g \in\mathcal{S}(\mathcal G_A\otimes\mathcal G_B)$ be states with bipartite structure. We say that $f$ \emph{locally contains} $g$ if there exists *-monomorphism $\pi_A:\mathcal G_A\to\mathcal F_A$ and $\pi_B:\mathcal G_B\to\mathcal F_B$ such that for all $G\in\mathcal G_A\otimes \mathcal G_B$
\[f\circ(\pi_A\otimes \pi_B)(G) = g(G).\]
\end{definition}
\begin{remark}
The fact that the overall monomorphism consists of tensor products of local monomorphism ensures that the local structure of $g$ across $\mathcal G_A$ and $\mathcal G_B$ is preserved in $f$. For example, if $g$ is an entangled state with Schmidt coefficients $\{a_i\}_i$, local homomorphisms will not change these Schmidt coefficients. Therefore, $f$ locally containing $g$ implies that $f$ contains a state with the same Schmidt coefficients as $g$. In this sense, $f$ contains a state that is equivalent to $g$ up to local unitary conjugations. This captures the situation where Alice and Bob share a state and are allowed to perform local unitaries on their respective subsystems.
\end{remark}

Our next step is to define what it means for a state $f$ to contian infinite copies of state $g$. We start by looking at what it means for $f$ to contain $2$ copies of $g$. A state $f\in\mathcal S(\mathcal F)$ contains $2$ copies of $g\in\mathcal S(\mathcal G)$ if there exists a $*$-isomorphism $\pi:\mathcal G\otimes \mathcal G \to\mathcal F$ such that $f\circ\pi = g\otimes g$.  Define $\pi_1:\mathcal G\to\mathcal F, \pi_1(a) = \pi(a\otimes 1)$ and $\pi_2:\mathcal G\to\mathcal F,\ \pi_2(a) = \pi(1\otimes a)$, then $f\circ \pi = g\otimes g$ is equivalent to $f(\pi_1(a_1)\pi_2(a_2)) = g(a_1) g(a_2)$. If we define $\mathcal F_1:=C^*(\pi_1(G))$ and $\mathcal F_2 := C^*(\pi_2(G))$, then we have the condition $f(x y) = f(x) f(y)$ for all $x\in\mathcal F_1$, $y\in\mathcal F_2$. Furthermore, from the definition of $\mathcal F_1$ and $\mathcal F_2$, we have $\mathcal F_1$ commutes with $\mathcal F_2$.

\begin{definition}[Infinite State Containment]\label{def:infinite_copies}
Let $ g \in \mathcal{S}(\mathcal{G}) $ be a state on a C*-algebra $ \mathcal{G} $, and let $ f \in \mathcal{S}(\mathcal{F}) $ be a state on a C*-algebra $ \mathcal{F} $. 

We say that: \emph{$ f $ contains infinitely many copies of $ g $} if there exists family of observable sets $ \{ \Sigma^{(i)} \subset \mathcal{F} \}_{i \in \mathbb{N}} $ such that:

\begin{enumerate}
    \item (\textbf{Certification}) For each $i$, there exits a *-isomorphism $\pi^{(i)}:\mathcal G\to\mathrm{C}^*(\Sigma^{(i)}) $ such that \footnote{We use $\mathrm{C}^*(S)$ of a set $S$ to denote the unital C*-algebra generated by the set $S$.},
    \[
    f\left(\pi^{(i)}(P)\right) = g\left(P \right)\quad \forall P\in\mathcal G.
    \].
    \item (\textbf{Mutual Commutativity}) For all $ i \ne j $, and all $O^{(i)}\in\Sigma^{(i)}$, $P^{(j)}\in\Sigma^{(j)}$,
    \[
    [O^{(i)}, P^{(j)}] = 0.
    \]
    \item (\textbf{Independence}) For any finite subset $N\subset \mathbb N$, $\{a_i: a_i\in\mathcal G\}_{i\in N}$, 
    \[f\left(\prod_{i\in N} \pi^{(i)}(a_i)\right) =\prod_{i\in N} g(a_i)\]
\end{enumerate}
\end{definition}

\begin{remark}
This definition captures the idea that the state $ f $ supports infinitely many non-interacting, certifiable simulations of the state $ g $. Each measurement set $ \Sigma^{(i)} $ reveals the same statistical structure as $ g $, and the mutual commutativity ensures that these measurements can be performed jointly without disturbance. The independence condition ensures that copies of $g$ are independent of each other and are not correlated. We note that this definition is slightly different from declaring $f$ contains $g^{\otimes\infty}$, but we will show that the two definitions are equivalent to each other. Definition~\ref{def:infinite_copies} captures the essence of what infinite copies of $g$ means without the need for another infinite-dimensional algebra to describe $g^\infty$. 
\end{remark}
\begin{proposition}\label{prop:infinite_copies_equivelance}
    Let $\mathcal{G}$ be a unital C*-algebra and $g \in \mathcal{S}(\mathcal{G})$ be a pure state. Define $\mathcal{C}$ to be the inductive limit of infinite minimal tensor product of $\mathcal G$,  $\mathcal C:= \bigotimes_{i\in\mathbb N} \mathcal{G}$\footnote{We use this $\bigotimes_{i\in\mathbb N} \mathcal G$ notation to denote $\varinjlim_n \bigotimes_{i=1}^n \mathcal G$, the completion of inductive limit of infinite tensor product.}, and define $g^{\infty} := g^{\otimes \infty}$ to be the unique product state on $\mathcal{C}$ such that for all $n \in \mathbb{N}$ and $a_1, \ldots, a_n \in \mathcal{G}$,
    \[
    g^{\infty}(a_1 \otimes \cdots \otimes a_n \otimes 1 \otimes 1 \cdots) = g(a_1)\cdots g(a_n).
    \]

    Then $f\in\mathcal S(\mathcal F)$ contains infinitely many copies of $g$ (in the sense of Definition~\ref{def:infinite_copies}) if and only if $f$ contains $g^{\infty}$ (in the sense of Definition~\ref{def:state_containment}).
\end{proposition}
\begin{proof}
    Suppose $f$ contains infinitely many copies of $g$, then we have a set of *-monomorphism $\{\pi_i: \mathcal G\to\mathcal F\}_{i\in\mathbb N}$. Define the *-monomorphism $\pi: \mathcal C\to\mathcal F$ as
        \[ \pi(a_1\otimes a_2\otimes\cdots \otimes a_n \otimes 1) = \prod_{i=1}^n \pi_i(a_i) \]
    Then 
    \[f(\pi(a_1\otimes\cdots \otimes a_n\otimes 1)) = \prod_{i=1}^n g(a_i) = g^\infty(a_1\otimes\cdots\otimes a_n\otimes 1).\]
    This is the definition for $f$ containing $g^\infty$.

         
    Now consider the case where $f$ contains $g^\infty$. Then there exists a *-monomorphism $\pi: \mathcal C\to\mathcal F$ such that $f\circ\pi = g^\infty$. We simply define $\pi^{(i)}(a) = \pi(1_{i-1}\otimes a \otimes 1)$ where $1_{i-1}$ is the identity on the first $i-1$ sites of $\mathcal C$. Then $\pi^{(i)}$ is a *-monomorphism from $\mathcal G$ to $\mathcal F$, and satisfies $f\circ\pi^{(i)} = g$. If we define $\mathcal F_i$ to be the image of $\pi_i$, it is a subalgebra of $\mathcal F$, and we get an *-isomorphism from $\mathcal G$ to $\mathcal F_i$. The commutativity relation is satisfied because each $\pi_i$ acts on different sites of $\mathcal C$, so $\mathcal F_i$ mutually commutes with each other. The independence relation is also satisfied directly from the definitions of $\pi^{(i)}$'s since $\prod_{i\in N} \pi^{(i)}(a_i) = \pi(b_1\otimes b_2\otimes\cdots)$ where $b_i = a_i$ if $i\in N$ and $b_i = 1$ otherwise.    
    Thus, $f$ contains infinitely many copies of $g$.
\end{proof}

Assume the state $g$ is split across two algebras, and we would like the containment of infinite copies of $g$ in $f$ to preserve the locality structure of $g$. We formally define this as local containment in the following sense.

\begin{definition}[Local Containment of Infinite copies]\label{def:local_infinite_containment}
Let \( f \in \mathcal{S}(\mathcal{F}) \) be a state on a C*-algebra \( \mathcal{F} \) admitting a local structure \( \mathcal{F}_A \otimes \mathcal{F}_B \subseteq \mathcal{F} \), and let \( g \in \mathcal{S}(\mathcal{G}_A \otimes \mathcal{G}_B) \) be a bipartite pure state.

We say that \( f \) \emph{locally contains infinitely many copies of} \( g \) if:
\begin{enumerate}
    \item \( f \) contains infinitely many copies of \( g \),
    \item For each \( n \), the *-isomorphism $\ \pi^{(n)}: \mathcal F_A\otimes\mathcal F_B \to \mathrm{C}^*(\Sigma^{(n)})$ can be factored as 
        \[ \pi^{(n)} = \pi_A^{(n)} \otimes \pi_B^{(n)}, \] where
    $\Sigma_A^{(n)} \subset \mathcal F_A$, $ \Sigma_B^{(n)} \subset \mathcal F_B,$
    and
    \[
    \pi_A^{(n)} : \mathcal G_A \to \mathrm{C}^*(\Sigma_A^{(n)}), \quad \pi_B^{(n)} : \mathcal G_B \to \mathrm{C}^*(\Sigma_B^{(n)})
    \]
    are *-isomorphisms.
\end{enumerate}
\end{definition}
This definition captures the idea that $f$ contains infinitely many copies of $g$ while preserving the locality structure of $g$. Each copy of $g$ is represented by a pair of mutually commuting observable sets, one contained in Alice's algebra and the other in Bob's. This ensures that the local operations and measurements that can be performed on each copy of $g$ are consistent with the original bipartite structure of $g$.

\begin{proposition}\label{prop:local_containment_infinite_copies} Use the definition of $f$ and $g$ in Def~\ref{def:local_infinite_containment}. Define $g^\infty := g^{\otimes\infty}$ as in Prop~\ref{prop:infinite_copies_equivelance} acting on the algebra $\mathcal C := \bigotimes_{i\in\mathbb N}\mathcal G_A\otimes\mathcal G_B$. Then $f$ locally contains infinite copies of $g$ if and only if $f$ locally contains $g^\infty$.
\end{proposition}
\begin{proof}
    Assuming $f$ locally contains infinite copies of $g$. Then there exists two internally mutually commuting sets of *-monomorphisms $\left\{\pi_A^{(i)}\right\}_i$ and $\left\{\pi_B^{(i)}\right\}_i$ such that $f\circ(\pi_A^{(i)}\otimes \pi_B^{(i)})(G) = g(G)$. Define $\pi_A:\bigotimes_{i\in\mathbb N}\mathcal G_A\to \mathcal F_A$ such that
    \[\pi_A (a_1\otimes a_2\otimes\cdots \otimes a_n\otimes 1) = \prod_{i=1}^n\pi_A^{(i)}(a_i)\]
    and similarly for $\pi_B$. Then using the argument in Prop~\ref{prop:infinite_copies_equivelance}, we have for any $n$, 
    \[ f\circ(\pi_A\otimes\pi_B)(a_1\otimes\cdots\otimes a_n\otimes 1 \otimes b_1\otimes\cdots\otimes b_n\otimes 1) = \prod_{i=1}^n g(a_i\otimes b_i)\]. Therefore we have $f\circ(\pi_A\otimes\pi_B) = g^\infty$, which is the definition for $f$ locally containing $g^\infty$.

    Now assuming $f$ locally contains $g^\infty$, then there exists $\pi_A\otimes\pi_B$ such that $f\circ(\pi_A\otimes \pi_B)(G) = g^\infty (G)$. From Prop~\ref{prop:infinite_copies_equivelance}, Condition~1 of Definition~\ref{def:local_infinite_containment} is satisfied. We just need to show Condition 2 is satisfied. Define $\pi_A^{(i)} (a)= \pi_A(1_{i-1}\otimes a\otimes 1)$ and $\pi_B^{(i)}(b) = \pi_B(1_{i-1}\otimes b\otimes 1)$, then $\pi_A^{(i)}$ and $\pi_B^{(i)}$ are $*$-monomorphisms from $\mathcal G_A\otimes \mathcal G_B$ to $\mathcal F_A\otimes\mathcal F_B$. Define $\pi^{(i)} = \pi_A^{(i)}\otimes \pi_B^{(i)}.$ This matches with the definition of $\pi^{(i)} (a_i) = \pi(1_{i-1}\otimes a_i\otimes 1)$ in Proposition~\ref{prop:infinite_copies_equivelance}, and therefore is a *-isomorphism that satisfies the non-locality condition.

\end{proof}

\subsection{Embezzlement Implies Local Containment of Infinite Copies}

Now we have our main theorem, which states that if $f$ can embezzle $g$, then $f$ must contain infinitely many copies $g$ that respects the locality of $g$.

\begin{theorem}\label{thm:emb_local_infinite_copies}
Let $ \mathcal F_A $, $\mathcal G_A$, $ \mathcal F_B $, $\mathcal G_B$ be unital C*-algebras, and let
\[
\Phi_A: \mathcal F_A \otimes \mathcal G_A \to \mathcal F_A, \quad \Phi_B: \mathcal F_B \otimes \mathcal G_B\to \mathcal F_B
\]
be *-isomorphisms, inducing
\[
\Phi := \Phi_A \otimes \Phi_B: \mathcal F_A \otimes \mathcal F_B \otimes \mathcal G_A\otimes\mathcal G_B \to \mathcal F_A \otimes \mathcal F_B
\]
with implicit re-ordering of $\mathcal G_A$ and $\mathcal F_B$. 

Let $ f \in \mathcal{S}(\mathcal F_A \otimes \mathcal F_B) $ be a state such that
\[
f \circ \Phi = f \otimes g
\]
for some entangled $ g \in \mathcal{S}(\mathcal G_A\otimes\mathcal G_B) $. Then $ f $ locally contains infinitely many copies of $ g $.
\end{theorem}

\begin{proof}
Let $ A\in \mathcal G_A $, $ B\in \mathcal G_B $ be observable sets that generate $\mathcal G_A$ and $\mathcal G_B$.
\[
\Sigma := \{A \otimes 1_{\mathcal G_B},\ 1_{\mathcal G_A} \otimes B \} \subset \mathcal G_A \otimes \mathcal G_B.
\]

Define, for $ n = 1 $,
\[
A^{(1)} := \Phi_A(1_{\mathcal F_A} \otimes A) \in \mathcal F_A, \quad
B^{(1)} :=  \Phi_B(1_{\mathcal F_B} \otimes B) \in \mathcal F_B,
\]
and recursively for $ n \ge 2 $,
\[
A^{(n)} := \Phi_A(A^{(n-1)} \otimes 1_{\mathcal G_A}), \quad
B^{(n)} := \Phi_B(B^{(n-1)}\otimes  1_{\mathcal G_B}).
\]

Then define the lifted observable set
\[
\Sigma^{(n)} := \left\{
A^{(n)}\otimes 1_{\mathcal F_B} ,\ 
 1_{\mathcal F_A}\otimes B^{(n)}\right\} \subset \mathcal F_A \otimes \mathcal F_B.
\]

We will show that the sequence $ \{ \Sigma^{(n)} \} $ satisfies:
\begin{itemize}
    \item \textbf{(Commutativity)} For all $ n \ne m $, and all $ O \in \Sigma^{(n)} $, $ O' \in \Sigma^{(m)} $, one has $ [O, O'] = 0 $.
    \item \textbf{(Certification)} For each $n$, there exists a *-isomorphism $\Phi^{(n)}:\mathcal G_A\otimes\mathcal G_B\to\mathrm{C}^*(\Sigma^{(n)})$ such that 
    \[
    f(\Phi^{(n)}(P)) = g(P)\quad \forall P\in \mathcal G_A\otimes\mathcal G_B.
    \]
    \item \textbf{(Independence)} For any finite $N\subset \mathbb N$, $P_i\in\mathcal G_A\otimes\mathcal G_B$, a collection of the *-isomorphism $\{\Phi^{(i)}\}_{i\in N}$ defined previously satisfies
    \[f\left(\prod_{i\in N} \Phi^{(i)}(P_i)\right) = \prod_{i\in N}g(P_i)\]
\end{itemize}

We start by defining the $\Phi_A^{(n)}:\mathcal F_A\otimes\mathcal G_A\to\mathcal F_A$ inductively as *-monomorphisms. 
\[\Phi_A^{(1)}(X)= \Phi_A(X)\]
and 
\[\Phi_A^{(n)}(X)= \Phi_A(\Phi_A^{(n-1)}(X)\otimes 1_{\mathcal G_A})\]
Then we have for $n \geq  1$,
\[\Phi_A^{(n)}(1_{\mathcal F_A}\otimes A) = A^{(n)}.\]
To show \textbf{commutativity}, note that  $A^{(n)} \otimes 1_{\mathcal F_B} $ trivially commutes with $1_{\mathcal F_A} \otimes B^{(m)}$. 
Hence, it suffices to show the commutativity between elements of $\{A^{(m)}:A\in\mathcal G_A\}$ and $\{A^{(n)}:A\in\mathcal G_A\}$; the same argument applies for $B$.

We first show that $A^{(1)}_x$ commutes with  $A_y^{(n)}$ for any $A_x, A_y\in\mathcal G_A$ and \(n > 1\). Observe that  $A_y^{(n-1)} \otimes 1_{\mathcal G_A}$commutes with $1_{\mathcal F_A} \otimes A_x^{(1)}$, and since \(\Phi_A\) is an isomorphism that preserves commutativity, it follows that
\[
A_x^{(1)} = \Phi_A(1\otimes A_x)\quad \text{commutes with} \quad A_y^{(n)}= \Phi_A(A^{(n-1)}_y \otimes 1_{\mathcal G_A}).
\]
Now, for \(m \neq n\), assume without loss of generality that \(m < n\). Then 
\[A^{(m)}_x = \Phi_A^{(m)}(1_{\mathcal F_A}\otimes A_x^{(1)}), \quad \text{and} \quad  A_y^{(n)} = \Phi_A^{(m)}(A_y^{(n-m)}\otimes 1_{\mathcal G_A}).\]
$\Phi_A^{(m)}$ is a *-homomorphism that preserves commutivity and $1_{\mathcal F_A} \otimes A_x^{(1)}$ commutes with $A_y^{(n-m)}\otimes 1_{\mathcal G_A}$, so $A_x^{(m)}$ commutes with $A_y^{(n)}$ for all $m\neq n$. 

Using the same argument, we can easily see that $B_x^{(m)}$ and $B_y^{(n)}$ commutes for all $m\neq n$. Therefore elements of $\Sigma^{(n)}$ commutes with elements of $\Sigma^{(m)}$ for all $m\neq n$. 

\smallskip

To show \textbf{certification}:
Define the following map
\[\tilde\Phi_A^{(n)}:\mathrm {C}^*(\Sigma)\to\mathrm{C}^*(\Sigma^{(n)}), \quad \tilde \Phi_A^{(n)} (X) = \Phi_A^{(n)}(1_{\mathcal F_A}\otimes X).\]
$\tilde\Phi_A^{(n)}$ is a *-isomorphism because it is one-to-one and its image is the very definition of the full algebra of $\mathrm{C}^*(\Sigma^{(n)})$. Moreover,  
\[\tilde\Phi_A^{(n)}(A^{(1)}) = A^{(n)}.\]
We define $\tilde \Phi_B^{(n)}$ similarly, then the combined map $\Phi^{(n)} = \tilde \Phi_A^{(n)}\otimes\tilde \Phi_B^{(n)}$ is a *-isomorphism between $\textrm{C}^*(\Sigma) =  \mathcal G_A\otimes\mathcal G_B$ and $\mathrm{C}^*(\Sigma^{(n)})$. We further have
\[\Phi^{(n)}(X) = \Phi_A\otimes\Phi_B(\Phi^{(n-1)}_A\otimes\Phi_B^{(n-1)}(X)\otimes 1_{\mathcal G_A}\otimes 1_{\mathcal G_B}) = \Phi(\Phi^{(n-1)}(X)\otimes 1_{\mathcal G_A}\otimes 1_{\mathcal G_B}).\]

Using the assumption $f\circ \Phi = f\otimes g$, we can show that for any $A\in\mathcal G_A$, $B\in\mathcal G_B$
\[f(\Phi^{(1)}(A\otimes B)) = f\circ\Phi(1_{\mathcal F_A}\otimes 1_{\mathcal F_B}\otimes A\otimes B) = f(1)g(A\otimes B) = g(A\otimes B).\]
Inductively,
\begin{eqnarray*}
    f(\Phi^{(n)}(A\otimes B))&=& f(\Phi(\Phi^{(n-1)}(A\otimes B)\otimes 1_{\mathcal G_A}\otimes 1_{\mathcal G_B}))\\
    &=& f\otimes g(\Phi^{(n-1)}(A\otimes B)\otimes 1_{\mathcal G_A}\otimes 1_{\mathcal G_B}) \\ 
    &=& f(\Phi^{(n-1)}(A\otimes B) ) g(1) = f(\Phi^{(n-1)}(A\otimes B))\\
    &=& \cdots = g(A\otimes B)
\end{eqnarray*}
Therefore, by linearity
\[
f(\Phi^{(n)}(P)) = g(P)\quad \forall P\in\mathcal G_A\otimes \mathcal G_B.
\]
To show \textbf{independence}, we start by considering the simpler case of $f(\Phi^{(i)}(P_i)\Phi^{(j)}(P_j))$ with $i > j$.
\begin{eqnarray*}
    f(\Phi^{(i)}(P_i)\Phi^{(j)}(P_j)) &=& f(\Phi(\Phi^{(i-1)}(P_i)\otimes 1)\Phi(\Phi^{(j-1)}(P_j)\otimes 1))\\
    &=&f(\Phi(\Phi^{(i-1)}(P_i)\Phi^{(j-1)}(P_j)\otimes 1))\\
    &=& f \otimes g (\Phi^{(i-1)}(P_i)\Phi^{(j-1)}(P_j)\otimes 1) \\
    &=& f(\Phi^{(i-1)}(P_i)\Phi^{(j-1)}(P_j)) = \cdots \\
    &=& f(\Phi^{(i-j+1)}(P_i)\Phi^{(1)}(P_j))
\end{eqnarray*}
Recall that $\Phi^{(1)}(X) =\tilde\Phi_A^{(1)}\otimes\tilde\Phi_B^{(1)}(X) = \Phi_A^{(1)}\otimes\Phi_B^{(1)}(1\otimes X) = \Phi(1\otimes X)$. We get
\begin{eqnarray*}
    f(\Phi^{(i)}(P_i)\Phi^{(j)}(P_j)) &=& f(\Phi(\Phi^{(i-j)}(P_i)\otimes 1)\Phi(1\otimes P_j))\\
    &=& f(\Phi(\Phi^{(i-j)}(P_i)\otimes P_j)) \\
    &=& f\otimes g(\Phi^{(i-1)}(P_i)\otimes P_j) \\
    &=& f(\Phi^{(i-1)}(P_i)) g(P_j)\\
    &=& g(P_i)g(P_j).
\end{eqnarray*}
For any finite $N\subset\mathbb N$, the commutivity condition allows us to write $\prod_{i\in N}\Phi^{(i)}(P_i)$ as a multiplication in any order. We arrange the $\Phi^{(i)}$'s in descending order of $i$, and apply the same argument as above. Then we can derive
\[f\left(\prod_{i\in N} \Phi^{(i)}(P_i)\right) = \prod_{i\in N} g(P_i).\]

Moreover, since the isomorphism $\Phi^{(n)}$ is defined as $\tilde \Phi_A^{(n)}\otimes\tilde\Phi_B^{(n)}$ where locality is preserved for all $n$, we have our result $f$ locally contains infinite copies of $g$.
\end{proof}

\begin{corollary}
    Let $\mathcal F_A$, $\mathcal F_B$ be unital C*-algebras, and $f\in \mathcal S(\mathcal F_A\otimes\mathcal F_B)$ be a state. Then $f$ can be used as the catalyst to embezzle the state $g\in\mathcal S(\mathcal G_A\otimes\mathcal G_B)$ if $f$ locally contains infinite copies of $g$. Moreover, a state that is infinite copies of $g$ suffices to embezzle $g$.
\end{corollary}
\begin{proof}
The first statement follows directly from  Theorem~\ref{thm:emb_local_infinite_copies} corresponds exactly to the situation where $f$ is the catalyst that embezzles $g$.

The second statement is an extension \cite{cleve2017perfect}, where an embezzlement protocol for Bell state is formulated using a catalyst consisting infinite copies of the Bell state with Hilbert-hotel intuition. Using the exact intution and formulation, to embezzle any target state $g$, the catalyst just needs to be changed into infinite copies of $g$ tensor infintie copies of $\ket{00}$..
\end{proof}

\section{Implications of Certifying Infinite Copies of $g$}

\subsection{Necessarity of Large Catalyst for Universal Embezzlement}
In \cite{liu2025explicit}, the explicit protocol of universal embezzlement is constructed by combining the embezzling catalyst of all different states into a gigantic catalyst state. In particular, the approximate universal embezzlement protocol is achieved by a catalyst that can exactly embezzle a dense set of all $n$-qubit states. One might question that if the state in this particular construction is unnecessarily large. Our result in this paper implies that, in order to exactly embezzle a countable dense set of all $n$-qubit states, the catalyst state in \cite{liu2025explicit} is in fact needed. 
\begin{proposition}\label{prop:necessity_large_catalyst}
    Let $s\in\mathcal S(\mathcal A\otimes\mathcal B)$ be a catalyst state for embezzlement that can embezzle any state of the form 
    \(\ket{\psi} = x_0\ket{00} + \cdots + x_n \ket{nn}\) with $x_i \in \mathbb Q^+$, 
    then $s$ must locally contain contably infinite copies of all states $\ket\psi$ with rational Schmidt coefficients simultaneously.
\end{proposition}
\begin{proof}
    Define 
    \[
      \Gamma = \{\ket\psi: \ket{\psi} = x_0\ket{00} + \cdots + x_n \ket{nn},\; x_i \in \mathbb Q^+\}.
    \]
    This is the set of all bipartite states with rational Schmidt coefficients. Since $\Gamma$ is countable, we may index it as $\Gamma=\{\ket{\psi_i}\}_{i\in\mathbb N}$.

    By Theorem~\ref{thm:emb_local_infinite_copies}, if $s$ can embezzle a state $g_i \in \mathcal S(\mathcal C_i \otimes \mathcal D_i)$ corresponding to $\ket{\psi_i}$, then $s$ must locally contain infinitely many copies of $g_i$. Concretely, there exist *-isomorphisms
\[
  \Phi_{Ai} : \mathcal A \otimes \mathcal C_i \to \mathcal A, 
  \qquad
  \Phi_{Bi} : \mathcal B \otimes \mathcal D_i \to \mathcal B
\]
such that 
\[
  s \circ (\Phi_{Ai} \otimes \Phi_{Bi}) = s \otimes g_i,
\]
where the tensor product on the right-hand side is understood as the standard extension of states to the product algebra.

Set $\Phi_A^0 = \Phi_{A0}$, $\Phi_B^0 = \Phi_{B0}$, and define recursively
\[
  \Phi_A^n = \Phi_{An} \circ (\Phi_A^{n-1} \otimes I), 
  \qquad
  \Phi_B^n = \Phi_{Bn} \circ (\Phi_B^{n-1} \otimes I).
\]
By induction, this yields
\[
  s \circ (\Phi_A^n \otimes \Phi_B^n) = s \otimes g_0 \otimes \cdots \otimes g_n.
\]
Since each $\Phi_A^n$ and $\Phi_B^n$ is a *-isomorphism, $s$ can embezzle the finite tensor product $g_0 \otimes \cdots \otimes g_n$, and therefore locally contains infinitely many copies of it. By Proposition~\ref{prop:infinite_copies_equivelance}, this is equivalent to
\[
  (g_0 \otimes \cdots \otimes g_n)^{\otimes \infty} \;\cong\; g_0^{\otimes \infty} \otimes \cdots \otimes g_n^{\otimes \infty}.
\]

To extend to the infinite case, define
\[
  \mathcal C = \bigotimes_{i\in\mathbb N} \mathcal C_i, 
  \qquad
  \mathcal D = \bigotimes_{i\in\mathbb N} \mathcal D_i
\]
as the inductive limit of the finite minimal tensor products. We then define *-isomorphisms
\[
  \Phi_A : \mathcal A \otimes \mathcal C \to \mathcal A, 
  \qquad
  \Phi_B : \mathcal B \otimes \mathcal D \to \mathcal B
\]
by requiring, for all $n \in \mathbb N$, $a \in \mathcal A$, $b \in \mathcal B$, and $c_i \in \mathcal C_i$, $d_i \in \mathcal D_i$ with $i \le n$,
\[
  \Phi_A(a \otimes c_1 \otimes \cdots \otimes c_n \otimes 1) = \Phi_A^n(a \otimes c_1 \otimes \cdots \otimes c_n),
  \quad
  \Phi_B(b \otimes d_1 \otimes \cdots \otimes d_n \otimes 1) = \Phi_B^n(b \otimes d_1 \otimes \cdots \otimes d_n).
\]
Equivalently, $\Phi_A$ and $\Phi_B$ are the inductive limits of the sequences $\{\Phi_A^n\}$ and $\{\Phi_B^n\}$, and the compatibility condition
\[
\Phi_A^m = \Phi_A^n|_{\mathcal A\otimes \bigotimes_{i\le m} \mathcal C_i}, \qquad 
\Phi_B^m = \Phi_B^n|_{\mathcal B\otimes \bigotimes_{i\le m} \mathcal D_i} \quad (m<n)
\]
is automatically satisfied by the recursive construction.

Let $g=\bigotimes_i g_i\in\mathcal S(\mathcal C\otimes\mathcal D)$. Then
\[
  s\circ(\Phi_A\otimes\Phi_B) = s\otimes g,
\]
so $s$ can embezzle $g$. By the main theorem, $s$ must locally contain infinitely many copies of $g$. Proposition~\ref{prop:infinite_copies_equivelance} then implies that $s$ locally contains $g^\infty$.

Finally, we identify $g^\infty$. Let $\mathcal F_i=\mathcal C_i\otimes\mathcal D_i$ and
\[
  \mathcal F = \bigotimes_{i\in\mathbb N} \mathcal F_i = \mathcal C\otimes \mathcal D.
\]
Define
\[
  \mathcal F^\infty = \bigotimes_{n\in\mathbb N}\mathcal F^{(n)},
  \qquad
  \mathcal F_i^\infty = \bigotimes_{n\in\mathbb N} \mathcal F_i^{(n)}.
\]
An elementary tensor in $\mathcal F^\infty$ has the form
\[
  x = \bigotimes_{n\in E} \bigotimes_{i\in 1_n} x_{n,i},
\]
where $E\subset\mathbb N$ is finite, $1_n\subset\mathbb N$ is finite for each $n$, and $x_{n,i}\in \mathcal F_i$ 
(all other legs are identity). Similarly, an elementary tensor in $\bigotimes_i \mathcal F_i^\infty$ has the form
\[
  y = \bigotimes_{i\in J} \bigotimes_{n\in E_i} y_{i,n},
\]
where $J\subset\mathbb N$ is finite, each $E_i\subset\mathbb N$ is finite, and $y_{i,n}\in \mathcal F_i$.

The correspondence between the two descriptions comes from reindexing coordinates by the bijection
\((n,i)\mapsto(i,n)\). Explicitly,
\[
  \Phi\!\left(\bigotimes_{n\in E}\ \bigotimes_{i\in 1_n} x_{n,i}\right)
  = \bigotimes_{i\in \bigcup 1_n}\ \bigotimes_{n\in E_i} x_{n,i},
  \qquad E_i=\{n:(n,i)\in E\times 1_n\}.
\]
This is a permutation of finitely many tensor factors. Since the tensor norm is symmetric, $\Phi$ is a *-isomorphism; its inverse reverses the permutation. Passing to the completion, we obtain
\[
  \mathcal F^\infty \;\cong\; \bigotimes_{i\in\mathbb N} \mathcal F_i^\infty.
\]
Using this *-isomorphism $\Phi$, we can transform $g^\infty$ into $\bigotimes_i (g_i^\infty)$. We can further decompose $\Phi$ as $\Phi_A^\infty\otimes\Phi_B^\infty$, where $\Phi_A^\infty$ and $\Phi_B^\infty$ act analogously to $\Phi$ but on Alice's and Bob's local systems, respectively. Importantly, the permutation preserves the bipartite structure because each $g_i$ is an entangled state across $\mathcal C_i \otimes \mathcal D_i$, so the rearrangement only permutes tensor factors without affecting the internal entanglement. Consequently, $g^\infty$ can be locally transformed into $\bigotimes_i (g_i^\infty)$, and $s$ therefore locally contains infinite copies of all $g_i$ simultaneously.
\end{proof}

\section{Conclusion and Discussions}
We have shown that exact embezzlement forces the catalyst state to exhibit a specific locality structure. There are some further implications of our results and open questions to be explored.

More concretely, we proved that to embezzle a given state $g$ exactly, the catalyst must locally contain infinitely many copies of $g$. Moreover, the Hilbert-hotel–style protocol of \cite{cleve2017perfect} allows a catalyst that is of the form of infinite copies of $g$ to implement the embezzlement protocol, showing that any exact single-state protocol can ultimately be reduced to this canonical construction. Thus, the protocol of \cite{cleve2017perfect} may be regarded as the \emph{reference} protocol for exact embezzlement.

On the universal side, the explicit construction of \cite{liu2025explicit} already requires tensor products of states living in a non-separable Hilbert space. Proposition~\ref{prop:necessity_large_catalyst} can be extended to show that any exact universal embezzler must locally contain infinitely many copies of each state from any countable family of target states, simultaneously. However, it remains unclear how to strengthen this argument to cover the full uncountable collection of all target states. A resolution of this issue would provide an alternative proof that exact universal embezzlement necessarily requires non-separable spaces.

All our arguments crucially use the \emph{exactness} of the embezzlement task. A natural open problem is to investigate whether approximate embezzlement can be seen as robust self-tests. Can approximate embezzlers be said to locally contain infinitely many copies of the target state approximately? What does it mean to claim $f$ locally contains infinite copies of $g$ approximately in this context? Standard fidelity measures are inadequate: if $g_1\neq g_2$, then infinite copies of $g_1$ and infinite copies of $g_2$ are orthogonal, even if $g_1$ and $g_2$ are arbitrarily close in norm. Thus, identifying an appropriate metric or operational criterion for approximation remains an essential challenge for extending the present theory.

\section{Acknowledgements}
This work was supported by the Novo Nordisk Foundation (grant NNF20OC0059939 ‘Quantum for Life’) and the European Research Council (grant agreement number 101078107 ‘QInteract’). I would like to thank William Slofstra for providing insights to the construction of the *-isomorphism in the paper. I would also like to thank Richard Cleve for proposing the question and Laura Mancinksa for helpful discussions on early ideas of this work. Finally I would like to thank Lauritz van Luijk for pointing out mistakes in the earlier version of this work.
\bibliography{reference}
\bibliographystyle{plain}

\appendix
\section{Proof of Proposition~\ref{prop:equivalence_embezzlement}}
For clarity, we prove the equivalence between standard and no-input embezzlement protocols in the case of 2-qubit target states. The proof can be easily generalized to higher dimensions.

\begin{proof}
\textbf{From Standard to No-Input:}

Assume we are given a standard embezzlement protocol with:
\begin{itemize}
    \item catalyst state $ \ket{\psi_c} \in \mathcal{H} $,
    \item local unitaries $ U_A, U_B: \mathcal{H} \otimes \mathbb{C}^2 \to \mathcal{H} \otimes \mathbb{C}^2 $,
\end{itemize}
such that
\[
(U_A\otimes 1_2) (U_B \otimes 1_2)( \ket{\psi_c} \otimes \ket{00} ) = \ket{\psi_c} \otimes \ket{\psi}.
\]
We will construct $V_A, V_B$ and $\ket{\phi_c}$ that satisfy the no-input embezzlement definition:
\[(V_A\otimes 1_2) V_B \ket{\phi_c} = \ket{\phi_c}\otimes \ket\psi\]
and $(V_A\otimes 1_2)V_B = (V_B\otimes 1_2)V_A$.

We write $\mathbb C^\infty$ for the infinite tensor product Hilbert space 
$\bigotimes_{n=1}^\infty (\mathbb C^2,\ket 0)$ in the sense of 
Guichardet, with reference vector $\ket 0$. Define the following unitaries  
\[W:\mathbb C^\infty \to\mathbb C^2\otimes\mathbb C^\infty\]
where for an orthonormal basis element $\ket e\in\mathbb C^\infty$, $W \ket e = \ket{e_0}\otimes\ket{e_{1+}}$, where $\ket{e_0}$ is the first qubit of $\ket e$ and $\ket{e_{1+}}$ is the rest of the qubits of $\ket e$, so $W$ pulls out the first qubit in $\mathbb C^\infty$. Define 
\[S:\mathbb C^2\otimes\mathbb C^\infty \to\mathbb C^\infty\otimes\mathbb C^2\]
where for any $\ket a\otimes\ket b \in\mathbb C^2 \otimes\mathbb C^\infty$, $S\ket a\otimes \ket b = \ket b\otimes \ket a$ is the swap operator. 

We use the notation $1_{n, X}$ to denote identity of dimension $n$, and $X\in\{A, B, AB\}$ to denote the space on Alice or Bob's side or both sides. We also use the notation $\mathbb C^2_A$, $\mathbb C^2_B$, $\mathbb C^\infty_A$, $\mathbb C^\infty_B$ to distinguish between the registers from Alice and Bob when needed. 

We further argue that $\mathbb C^m_A \otimes \mathbb C^n_B$ is physically equivalent to $\mathbb C^n_B \otimes \mathbb C^m_A$. Mathematically, these two tensor products are not identical, since a swap operator is required to align the register order. Physically, however, they both describe the same non-local system composed of Alice's subsystem $\mathbb C^m_A$ and Bob's subsystem $\mathbb C^n_B$. The choice of ordering across the bipartition is therefore a matter of convention without physical significance. By contrast, $\mathbb C^m_A \otimes \mathbb C^n_A$ is not interchangeable with $\mathbb C^n_A \otimes \mathbb C^m_A$, because both factors belong to Alice's laboratory, and the ordering may correspond to the physical arrangement of qubits. In conclusion, the ordering of registers \emph{across} Alice and Bob's local systems is irrelevant, while the ordering of registers \emph{within} one party's system can be operationally meaningful. We therefore interpret reordering across the bipartition as a null operation and write
\[
   \mathbb C^m_A \otimes \mathbb C^n_B \;\equiv\; \mathbb C^n_B \otimes \mathbb C^m_A.
\]

Define the following swap unitaries
\[S_A = 1_{\mathcal H}\otimes S\otimes 1_{\infty,B},\quad S_B  = 1_\mathcal H\otimes 1_{\infty, A}\otimes S\]
where $S_A:\mathcal H\otimes \mathbb C^2_A\otimes \mathbb C^\infty_A\otimes\mathbb C^\infty_B\to \mathcal H\otimes \mathbb C^\infty_A\otimes \mathbb C^2_A\otimes\mathbb C^\infty_B$ swaps the single-qubit register in Alice's system with the Alice's $\mathbb C^\infty$, and similarly, $S_B: \mathcal H\otimes \mathbb C_A^\infty \otimes \mathbb C^2_B\otimes\mathbb C_B^\infty \to \mathcal H \otimes \mathbb C^\infty_A\otimes \mathbb C^\infty_B\otimes\mathbb C^2_B$.

Define the local pull-out isometries 
\[W_A = 1_{\mathcal H}\otimes W\otimes 1_{\infty, B},\qquad W_B = 1_\mathcal H\otimes 1_{\infty, A}\otimes W,\]
where $W_A:\mathcal H\otimes \mathbb C^\infty_A\otimes\mathbb C^\infty_B\to \mathcal H\otimes \mathbb C^2_A\otimes \mathbb C^\infty_A\otimes \mathbb C^\infty_B$ pulls out a qubit from Alice's $\mathbb C^\infty$, and $W_B:\mathcal H\otimes \mathbb C^\infty_A\otimes \mathbb C^\infty_B\to \mathcal H\otimes \mathbb C^\infty_A\otimes \mathbb C^2_B\otimes \mathbb C^\infty_B$ pulls out a qubit from Bob's $\mathbb C^\infty$.

Define the embzzlement unitaries 
\[\tilde U_A = U_A\otimes 1_{\infty, AB}, \qquad \tilde U_B = U_B\otimes 1_{\infty, AB}\]
where $U_A$ acts on $\mathcal H\otimes\mathbb C^2_A$ in $\mathcal H\otimes \mathbb C_A^2\otimes\mathbb C^\infty_A\otimes\mathbb C^\infty_B$ and $U_B$ acts on $\mathcal H\otimes \mathbb C^2_B$ in $\mathcal H\otimes \mathbb C^\infty_A\otimes\mathbb C^2_B\otimes\mathbb C^\infty_B$.

Finally, define Alice and Bob's local isometries
\[V_A = S_A \tilde U_A W_A,\qquad V_B = S_B \tilde U_B W_B\]
Our catalyst state is 
\[\ket{\phi_c} = \ket{\psi_c}\otimes\ket{00}^\infty\]
We will show that
\[(V_A\otimes 1_{2, B}) V_B \ket{\phi_c} \equiv \ket{\phi_c}\otimes\ket\psi.\]
We first need to analyze how each element of $(V_A\otimes 1_{2, B})$ gets passed through elements of $V_B$.
\begin{eqnarray*}
    (W_A\otimes 1_{2, B} )S_B& =& 1_\mathcal H \otimes W \otimes S = (1_\mathcal H \otimes 1_{2, A}\otimes 1_{\infty, A}\otimes S) (1_\mathcal H\otimes 1_{2, B}\otimes W\otimes 1_{\infty, B}) =: \tilde S_B \tilde W_A\\
    (\tilde U_A\otimes 1_{2, B}) \tilde S_B &=& U_A\otimes 1_{\infty, A}\otimes S = \tilde S_B (U_A\otimes 1_{\infty, A} \otimes 1_{2, B}\otimes 1_{\infty, B}) =: \tilde S_B \hat U_A\\
    \tilde W_A \tilde U_B &=&  U_B\otimes W \otimes 1_{\infty, B}= (U_B\otimes 1_{2, A}\otimes 1_{\infty, AB})\tilde W_A=:\hat U_B \tilde W_A\\
    \tilde W_A W_B &=&1_{\mathcal H}\otimes W\otimes W\\
    \hat U_A \hat U_B &=& (U_A\otimes 1_{\infty, A}\otimes 1_{2, B}\otimes 1_{\infty, B})(U_B\otimes 1_{2, A}\otimes 1_{\infty, AB})\\
    (S_A\otimes 1_{2, B})\hat S_B &=& 1_{\mathcal H}\otimes S\otimes S
\end{eqnarray*} 
Combining all of them together, we have 
\begin{eqnarray*}
    (V_A\otimes 1_{2, B}) V_B \ket{\phi_c} &=& (S_A\otimes 1_{2, B})(\tilde U_A\otimes 1_{2, B})(W_A\otimes 1_{2, B}) S_B \tilde U_B W_B \ket{\psi_c}\otimes\ket{00}^\infty\\
    &=& ((S_A\otimes 1_{2,B})\tilde S_B) (\hat U_A\hat U_B )(\tilde W_A W_B )\ket{\psi_c}\otimes\ket{00}^\infty\\
    &=& ((S_A\otimes 1_{2,B})\tilde S_B) (\hat U_A\hat U_B ) (1_\mathcal H\otimes W\otimes W)\ket{\psi_c}\otimes \ket{00}^\infty\\
    &=&  ((S_A\otimes 1_{2,B})\tilde S_B)\hat U_A\hat U_B\ket{\psi_c}\otimes \ket{0}\otimes\ket 0^\infty\otimes \ket{0}\otimes \ket{0}^\infty\\
    &=& (1_\mathcal H\otimes S\otimes S) (1_\mathcal H\otimes 1_{2,A}\otimes S \otimes 1_{\infty, B})\ket{\psi_c}\otimes \ket\psi\otimes\ket{00}^\infty\\
    &\equiv& \ket{\psi_c}\otimes\ket{00}^\infty\otimes\ket\psi = \ket{\phi_c}\otimes\ket\psi
\end{eqnarray*}
To show that $(V_A\otimes 1_2)V_B = (V_B\otimes 1_2)V_A$, we similarly analyze how each element of $(V_B\otimes 1_{2, A})$ gets passed through elements of $V_A$, and get
\[(V_B\otimes 1_{2, A}) V_A = ((S_B\otimes 1_{2, B})\hat S_A)\hat U_B \hat U_A \tilde W_B W_A = (1_\mathcal H \otimes S\otimes S) \hat U_B\hat U_A (1_\mathcal H\otimes W\otimes W)\]
and since $\hat U_A\hat U_B = \hat U_B\hat U_A$ from the commutivity of $U_A$ and $U_B$, we have $(V_A\otimes 1_2)V_B = (V_B\otimes 1_2)V_A$.
\smallskip

\textbf{From No-Input to Standard:}

Assume we are given:
\begin{itemize}
    \item a catalyst state $ \ket{\phi_c} \in \mathcal{H} $,
    \item local isometries $ V_A, V_B: \mathcal{H} \to \mathcal{H} \otimes \mathbb{C}^2 $ satisfying
    \[
    (V_A \otimes 1_2)(V_B \ket{\phi_c}) = \ket{\phi_c} \otimes \ket{\psi}.
    \]
     as well as $(V_A \otimes 1_2)V_B = (V_B \otimes 1_2)V_A$.
\end{itemize}
Our goal is to construct $U_A, U_B$ and $\ket{\psi_c}$ such that 
\[(U_A\otimes 1_{2, B})(U_B\otimes 1_{2, A}) \ket{\psi_c}\otimes\ket{00} = \ket{\psi_c}\otimes\ket{\psi}\]
satisfying $(U_A\otimes 1_{2, B})(U_B\otimes 1_{2, A}) = (U_B\otimes 1_{2, A})(U_A\otimes 1_{2, B})$.

We start with the defining $\ket{\psi_c} = \ket{\phi_c}\otimes\ket{00}^\infty$. 

We keep the definition of $S$ and re-define $W$ as the following push-in isometry:
\[W:\mathbb C^\infty\otimes \mathbb C^2\to\mathbb C^\infty\]
where $W\ket{e}\otimes \ket{a} = \ket{e_a}$ where $\ket{e_a}$ is the element in $\mathbb C^\infty$ whose first qubit is $\ket a$ and the rest of the qubits are the same as $\ket e$.

Define $W_A$, $W_B$, $S_A$, $S_B$ as before, but with the new definition of $W$. Define $U_A$ and $U_B$ as follows:
\[U_A = S_A (V_A\otimes 1_{\infty, AB}) W_A, \quad U_B = S_B(V_B\otimes 1_{\infty, AB}) W_B\]
Again, we analyze how each element of $(U_A\otimes 1_{2, B})$ gets passed through elements of $U_B\otimes 1_{2, A}$.
\begin{eqnarray*}
  (W_A\otimes 1_{2, B}) (S_B\otimes 1_{2, A})
  &\equiv& 1_\mathcal H\otimes W\otimes S = (1_\mathcal H\otimes 1_{\infty, A}\otimes S)(1_{\mathcal H} \otimes W \otimes 1_{2, B}\otimes 1_{\infty, B}) =: S_B \tilde W_A\\
  \tilde W_A(V_B\otimes 1_{\infty, AB} \otimes 1_{2, A}) &\equiv&
  V_B \otimes W \otimes 1_{\infty, B} = (V_B\otimes 1_{\infty, A}\otimes 1_{\infty, B}) W_A =: \hat U_B  W_A\\
  (V_A\otimes 1_{\infty, AB}\otimes 1_{2, B}) S_B &= & V_A\otimes 1_{\infty, A}\otimes S = (1_\mathcal H \otimes 1_{2, A}\otimes 1_{\infty, A}\otimes S) (V_A\otimes 1_{\infty, A}\otimes 1_{2, B}\otimes 1_{\infty, B}) =: \tilde S_B \hat U_A\\
  (S_A\otimes 1_{2, B}) \tilde S_B &=& 1_\mathcal H\otimes S\otimes S\\
  \tilde W_A (W_B\otimes 1_{2, A}) &\equiv& 1_\mathcal H\otimes W\otimes W
\end{eqnarray*}
We note the use of $\equiv$ involving $X\otimes 1_{2, A}$ because while we used $U_B\otimes 1_{2, A}\in \mathcal U(\mathcal H\otimes \mathbb C^\infty_{AB}\otimes \mathbb C^2_A)$, we interpret it as an unitary acting on $\mathcal H\otimes \mathbb C^\infty_A\otimes \mathbb C^2_A\otimes \mathbb C^\infty_B$.

Then  we have the commutation relation
\begin{eqnarray*}
  (U_A\otimes 1_{2, B})(U_B\otimes 1_{2, A}) &\equiv& (1_\mathcal H\otimes S\otimes S)
(V_A\otimes 1_{\infty, A}\otimes 1_{2, B}\otimes 1_{\infty, B})(V_B\otimes 1_{\infty, AB})(1_\mathcal H\otimes W\otimes W)\\
&=& (1_\mathcal H\otimes S\otimes S)
(V_B\otimes 1_{2, A}\otimes 1_{\infty, A}\otimes 1_{\infty, B})(V_A\otimes 1_{\infty, AB})(1_\mathcal H\otimes W\otimes W)\\
&\equiv& (U_B\otimes 1_{2, A})(U_A\otimes 1_{2, B})
\end{eqnarray*}
Now to analyze the overall action of $(U_A\otimes 1_{2, B})(U_B\otimes 1_{2, A})$ on $\ket{\psi_c}\otimes\ket{00}$, we have
\[\ket{\psi_c}\otimes\ket{00} = \ket{\phi_c}\otimes\ket{00}^\infty\otimes\ket{00} \equiv \ket{\phi_c}\otimes\ket{0}^\infty_A\otimes\ket{0}_A\otimes\ket 0^\infty_B\otimes\ket 0_B\]
which allows us to compute
\begin{eqnarray*}
  (I\otimes W\otimes W ) \ket{\phi_c}\otimes \ket 0^\infty_A\otimes \ket 0_A\otimes \ket 0^\infty_B\otimes \ket 0_B&=& \ket{\phi_c}\otimes\ket{00}^\infty\\
  (V_A\otimes 1_{\infty, A}\otimes 1_{2, B}\otimes 1_{\infty, B})(V_B\otimes 1_{\infty, AB}) \ket{\psi_c}\otimes\ket{00}^\infty &=& \ket{\psi_c}\otimes \ket{\psi}\otimes\ket{00}^\infty \\
  &\equiv& (1_{\mathcal H}\otimes 1_{2,A}\otimes S \otimes 1_{\infty, B})\ket{\psi_c}\otimes \ket{\psi}\otimes\ket{00}^\infty\\
  (I\otimes S\otimes S)  (1_{\mathcal H}\otimes 1_{2,A}\otimes S \otimes 1_{\infty, B})\ket{\phi_c}\otimes \ket\psi\otimes \ket{00}^\infty &\equiv & \ket{\phi_c}\otimes\ket{00}^\infty\otimes \ket{\psi} = \ket{\psi_c}\otimes\ket\phi
\end{eqnarray*}
\end{proof}

\end{document}